\documentclass{article}

\usepackage{amsfonts}
\usepackage{amssymb}
\usepackage{hyperref}
\usepackage[nosumlimits,nointlimits,nonamelimits]{amsmath}
\usepackage{amsmath}
\usepackage{graphicx}

\newcommand{\hv}{h.v.\thinspace}

\newtheorem{theorem}{Theorem}[section]
\newtheorem{proposition}[theorem]{Proposition}
\newtheorem{definition}[theorem]{Definition}
\newtheorem{lemma}[theorem]{Lemma}

\newtheorem{remark}[theorem]{Remark}
\newtheorem{corollary}[theorem]{Corollary}

\newenvironment{proof}[1][Proof]{\noindent \textbf{#1.} }{\  \rule{0.5em}{0.5em}}

\begin{document}

\title{Fiber Products of Measures and Quantum Foundations\thanks{We are grateful to Samson Abramsky, Bob Coecke, Amanda Friedenberg, Barbara Rifkind, Gus Stuart, and Noson Yanofsky for valuable conversations, to John Asker, Axelle Ferri\`ere, Elliot Lipnowski, Andrei Savochkin, participants at the workshop on Semantics of Information, Dagstuhl, June 2010, and participants at the conference on Advances in Quantum Theory, Linnaeus University, V\"{a}xj\"{o}, June 2010, for helpful input, to a referee and the volume editors for very important feedback, and to the Stern School of Business for financial support. {\tiny fpmqf-01-14-14}}}
\author{Adam Brandenburger\thanks{Address: Stern School of Business, New York University, New York, NY 10012, \href{mailto:adam.brandenburger@stern.nyu.edu}{adam.brandenburger@stern.nyu.edu}, \href{http://stern.nyu.edu/\~abranden}{www.stern.nyu.edu/$\sim$abranden}}
\and
H. Jerome Keisler\thanks{Address: Department of Mathematics, University of Wisconsin-Madison, Madison, WI 53706, \href{mailto:keisler@math.wisc.edu}{keisler@math.wisc.edu}, \href{http://math.wisc.edu/\~keisler}{www.math.wisc.edu/$\sim$keisler}}}
\date{Version 01/14/14\\
\medskip
}
\maketitle

\begin{abstract}
With a view to quantum foundations, we define the concepts of an \textbf{empirical model} (a probabilistic model describing measurements and outcomes), a \textbf{hidden-variable model} (an empirical model augmented by unobserved variables), and various properties of hidden-variable models, for the case of infinite measurement spaces and finite outcome spaces.  Thus, our framework is general enough to include, for example, quantum experiments that involve spin measurements at arbitrary relative angles.  Within this framework, we use the concept of the \textbf{fiber product of measures} to prove general versions of two determinization results about hidden-variable models.  Specifically, we prove that: (i) every empirical model can be realized by a deterministic hidden-variable model; (ii) for every hidden-variable model satisfying locality and $\lambda$-independence, there is a realization-equivalent hidden-variable model satisfying determinism and $\lambda$-independence.
\end{abstract}

\section{Introduction}
\label{s-intro}
\thispagestyle{empty}
Hidden variables are extra variables added to the model of an experiment to explain correlations in the outcomes.  Here is a simple example.  Alice's and Bob's computers have been prepared with the same password.  We know that the password is either p2s4w6r8 or 1a3s5o7d, but we do not know which it is.  If Alice now types in p2s4w6r8 and this unlocks her computer, we immediately know what will happen when Bob types in one or other of the two passwords.  The two outcomes --- when Alice types a password and Bob types a password --- are perfectly correlated.  Clearly, it would be wrong to conclude that, when Alice types a password on her machine, this somehow causes Bob's machine to acquire the same password.  The correlation is purely informational: It is our state of knowledge that changes, not Bob's computer.  Formally, we can consider an r.v. (random variable) $X$ for Alice's password, an r.v.\,\,$Y$ for Bob's password, and an extra r.v.\,\,$Z$.  The r.v.\,\,$Z$ takes the value $z_1$ or $z_2$ according as the two machines were prepared with the first or the second password.  Then, even though $X$ and $Y$ will be perfectly correlated, they will also be independent (trivially so), conditional on the value of $Z$.  In this sense, the extra r.v.\,\,$Z$ explains the correlation.

Of course, even in the classical realm, there are much more complicated examples of hidden-variable analysis.  But, the most famous context for hidden-variable analysis is quantum mechanics (QM).  Having started with von Neumann \cite[1932]{vonneumann32} and Einstein, Podolosky, and Rosen \cite[1935]{einstein-podolosky-rosen35}, the question of whether a hidden-variable formulation of QM is possible was re-ignited by Bell \cite[1964]{bell64}, whose watershed no-go theorem gave conditions under which the answer is negative.  The correlations that arise in QM --- for example, in spin measurements --- cannot be explained as reflecting the presence of hidden variables.

Let us specify a little more what we mean by an experiment.  We imagine that Alice can make one of several measurements on her part of a certain system, and Bob can make one of several measurements on his part of the system.  Each pair of measurements (one by Alice and one by Bob) leads to a pair of outcomes (one for Alice and one for Bob).  We can build an \textbf{empirical model} of the experiment by choosing appropriate spaces for the sets of possible measurements and outcomes, and by specifying, for each pair of measurements, a probability measure over pairs of outcomes.  An associated \textbf{hidden-variable} (henceforth \textbf{h.v.}) \textbf{model} is obtained by starting with the empirical model and then appending to it an extra r.v..

We can define various types of \hv model, according to what properties we ask of the model.  One property is \textbf{locality} (Bell \cite[1964]{bell64}), which can be decomposed into \textbf{parameter independence} and \textbf{outcome independence} (Jarrett \cite[1984]{jarrett84}, Shimony \cite[1986]{shimony86}).  Another property is \textbf{$\lambda$-independence} (the term is due to Dickson \cite[2005]{dickson05}), which says that the choices of measurement by Alice and Bob are independent of the process determining the values of any h.v.s.  Bell \cite[1985, p.95]{bell85} describes this as the condition that ``the settings of instruments are in some sense free variables."  We will use the term ``free variables" below.

Here are two basic types of \hv question one can ask:

\begin{itemize}

\item[i.]\textbf{The existence question}\,\,\,\,Suppose we are given a certain physical system and an empirical probability measure $e$ on the observable variables of the system.  Can we find an extended space that includes h.v.s, and a probability measure $p$ on this space, where $p$ satisfies certain properties (as above) and realizes (via marginalization) the empirical probability measure $e$?

\item[ii.]\textbf{The equivalence question}\,\,\,\,Suppose we are given an empirical probability measure $e$ on the observable variables of a system, and an \hv model, with probability measure $p$ that satisfies certain properties and realizes $e$.  Can we find another \hv model, with probability measure $q$, where $q$ satisfies other stipulated properties and also realizes $e$?

\end{itemize}

Bell's Theorem is the most famous negative answer to {i.}, obtained when the physical system is quantum and the properties demanded are locality and $\lambda$-independence.

In this paper we will focus on positive results for questions of both types {i.} and ii.  These positive results involve yet another property of \hv models:  The (strong) \textbf{determinism} property says that for each player, the h.v.s\, determine \lq non-probabilistically\rq\,\,(formally:\,\,almost surely) the outcome of any measurement.  As we will see in Section \ref{s-properties}, determinism implies locality.  We consider the following positive results on questions {i.} and {ii.}:

\begin{itemize}

\item[i.]\textbf{First determinization result} Every empirical model (whether generated by a classical or quantum or even superquantum system) can be realized by an \hv model satisfying determinism.

\item[ii.]\textbf{Second determinization result}\,\,Given an \hv model satisfying locality and $\lambda$-independence, there is a realization-equivalent \hv model that satisfies determinism and $\lambda$-independence.

\end{itemize}

Put together, these two results tell us a lot about Bell's Theorem.   The first determinization result says that for every empirical model, an \hv model with determinism is possible.  It is also true that for every empirical model, an \hv model with $\lambda$-independence is possible.  (This is a trivial construction, which we note in Remark \ref{r-singleton}.)   As usually stated, Bell's Theorem asks for an \hv model satisfying locality and $\lambda$-independence.  In light of the second determinization result, Bell's Theorem can be equivalently stated as asking for determinism and $\lambda$-independence.  Thus, Bell's Theorem teaches us that: \textit{It is possible to believe that Nature (in the form of QM) is deterministic, or it is possible to believe that measurement choices by experimenters are free variables, but it is not possible to believe both.}

The goal of this paper is to prove the two determinization results at a general measure-theoretic level (Theorems \ref{t-firsttheorem} and \ref{t-secondtheorem}).  Bell \cite[1971]{bell71} mentioned the idea of the first determinization result.  Fine \cite[1982]{fine82} produced the first version of the second determinization result.  Both results have been (re-)proved for various formulations in the literature.  A notable aspect of our formulation is that we allow for infinite measurement spaces.  Thus, our set-up is general enough to include, for example, experiments that involve spin measurements at arbitrary relative angles.  We assume that outcome sets are finite (such as spin up or spin down).

Our treatment uses the concept of the \textbf{fiber product of measures}.  The construction of these objects comes from Shortt \cite[1984]{shortt84}.  The name ``fiber product" is taken from Ben Yaacov and Keisler \cite[2009]{benyaacov-keisler09}, who employed the concept in the context of continuous model theory.  Fiber products of measures turn out to be well suited to the questions in quantum foundations which we study in this paper.

\section{Empirical and Hidden-Variable Models}
\label{s-empirical}
Alice has a space of possible measurements, which is a measurable space $(Y_{a},\mathcal{Y}_{a})$, and a space of possible outcomes, which is a measurable space $(X_{a},\mathcal{X}_{a})$.  Likewise, Bob has a space of possible measurements, which is a measurable space $(Y_{b},\mathcal{Y}_{b})$, and a space of possible outcomes, which is a measurable space $(X_{b},\mathcal{X}_{b})$.  Throughout, we will restrict attention to bipartite systems.  (We will comment later on the extension to more than two parts.)  There is also an \hv space, which is an unspecified measurable space $(\Lambda,\mathcal{L})$.  Write
\begin{align*}
(X,\mathcal{X}) & =  (X _{a},\mathcal{X}_{a})\otimes(X _{b},\mathcal{X}_{b}), \\
(Y,\mathcal{Y}) & =  (Y_{a},\mathcal{Y}_{a})\otimes(Y_{b},\mathcal{Y}_{b}), \\
\Psi & =  (X,\mathcal{X})\otimes(Y,\mathcal{Y}), \\
\Omega & =  (X,\mathcal{X})\otimes(Y,\mathcal{Y})\otimes(\Lambda ,\mathcal{L}).
\end{align*}

\begin{definition}
An \textbf{empirical model} is a probability measure $e$ on $\Psi$.
\end{definition}

We see that an empirical model describes an experiment in which the pair of measurements $y=(y_{a},y_{b})\in Y$ is randomly chosen according to the probability measure ${\rm marg}_{Y}e$, and $y$ and the joint outcome $x=(x_{a},x_b)\in X$ are distributed according to $e$.

\begin{definition}
A \textbf{hidden-variable} (\textbf{h.v.}) \textbf{model} is a probability measure $p$ on $\Omega $.
\end{definition}

\begin{definition}
We say that an \hv model $p$ \textbf{realizes} an empirical model $e$ if $e={\rm marg}_{\Psi}p$.  We say that two \hv models, possibly with different \hv spaces, are (\textbf{realization-})\textbf{equivalent} if they realize the same empirical model.
\end{definition}

An \hv model is an empirical model which has an extra component, viz., the \hv space, and which reproduces a given empirical model when we average over the values of the h.v..  The interest in \hv models is that we can ask them to satisfy properties that it would be unreasonable to demand of an empirical model.  Thus, in the example we began with, the property we ask for is conditional independence --- which we would only expect once the extra r.v.\,\,$Z$ is introduced.  We will come to other properties in Section \ref{s-properties}.

\section {Products and Fiber Products of Measures}
\label{s-preliminaries}
We first introduce notation and recall some well-known facts about product measures.  For background on the relevant measure theory, see e.g.\,Billingsley \cite[1995]{billingsley95}.

Recall that by a product $(X,\mathcal{X})\otimes(Y,\mathcal{Y})$ of two measurable spaces $(X,\mathcal{X})$ and $(Y,\mathcal{Y})$ is meant the (Cartesian) product space $X \times Y$ equipped with the $\sigma$-algebra generated by the measurable rectangles $J \times K$, where $J \in \mathcal{X}$ and $K \in \mathcal{Y}$.  We use the following two conventions.  First, when $p$ is a probability measure on $(X,\mathcal{X})\otimes(Y,\mathcal{Y})$ and $q = {\rm marg}_{X}p$, then for each $J\in \mathcal{X}$ we write
\begin{equation*}
p(J) = p(J\times Y) = q(J),
\end{equation*}
and for each $q$-integrable $f:X\rightarrow\mathbb{R}$ we write
\begin{equation*}
\int_{J}f(x)\,dp = \int_{J\times Y}f(x)\,dp = \int_{J}f(x)\,dq.
\end{equation*}
Thus, in particular, a statement holds for $p$-almost all $x\in X$ if and only if it holds for $q$-almost all $x\in X$.

Second, when $p$ is a probability measure on a product space $(X,\mathcal{X})\otimes(Y,\mathcal{Y})\otimes(Z,\mathcal{Z})$, $J\in\mathcal{X}$, and $z \in Z$, we write $p[J||\mathcal{Z}]$ for the conditional probability of $J$ given $z$.  Here, we refer to the concept of conditional probability given a sub $\sigma$-algebra; see Billingsley \cite[1995, Section 33]{billingsley95} for a presentation.  Formally, $p[J||\mathcal{Z}]$ denotes a function from $Z$ into $[0,1]$ such that
\begin{equation*}
p[J||\mathcal{Z}]_{z} = p[J\times Y\times Z|\{X\times Y,\emptyset\}\otimes \mathcal{Z}]_{(x,y,z)}.
\end{equation*}
(Note that $\{X\times Y,\emptyset\}$ is the trivial $\sigma$-algebra over $X\times Y$, so that the right-hand side does not depend on $(x,y)$.)

We use similar notation for (finite) products with factors to the left of $(X,\mathcal{X})$ or to the right of $(Z,\mathcal{Z})$.  Note that if $q = {\rm marg}_{X\times Z}p$, then $q[J||\mathcal{Z}]=p[J||\mathcal{Z}]$.  We will also need the concept of conditional expectation given a sub $\sigma$-algebra (Billingsley \cite[1995, Section 34]{billingsley95}), and we will use an analogous notation.  Thus, given an integrable function $f:X\rightarrow\mathbb{R}$, and $z \in Z$, we define ${\rm E}[f||\mathcal{Z}]$ by:
\begin{equation*}
{\rm E}[f||\mathcal{Z}]_{z} = {\rm E}[f\circ\pi|\{X\times Y,\emptyset\}\otimes\mathcal{Z}]_{(x,y,z)},
\end{equation*}
where we write $\pi$ for the projection from $X\times Y\times Z$ to $X$.

\begin{lemma}
\label{l-cond}
The mapping $z\mapsto p[J||\mathcal{Z}]_{z}$ is the $p$-almost surely unique $\mathcal{Z}$-measurable function $f:Z\rightarrow [0,1]$ such that for each set $L\in\mathcal{Z}$,
\begin{equation*}
\int_{L}f(z)\,dp = p(J\times L).
\end{equation*}
\end{lemma}

\begin{proof}  Existence:
 Let $f(z) = p[J||\mathcal{Z}]_{z}$.  Using the definition of $p[J||\mathcal{Z}]$, we see that
\begin{equation*}
\int_{L}f(z)\,dp = \int_{X\times Y\times L}{\rm E}[1_{J\times Y\times Z}|\{X\times Y,\emptyset \}\otimes\mathcal{Z}]\,dp =
\end{equation*}
\begin{equation*}
\int_{X\times Y\times L}1_{J\times Y\times Z}\,dp = p((X\times Y\times L)\cap (J\times Y\times Z)) = p(J\times L),
\end{equation*}
as required.

Uniqueness:  If $p(J)=0$, then $f(z)=g(z)=0$ $p$-almost surely.  Suppose $p(J)>0$.
Let $f$ and $g$ are two such functions and let $L=\{z:f(z)<g(z)\}$.  Then $L\in\mathcal{Z}$.
If $p(J\times L)>0$, then $p(L)>0$, and
$$0<\int_L g(z)\,dp-\int_L f(z)\,dp=\int_L g(z)-f(z)\,dp=0,$$
a contradiction.  Therefore $p(J\times L)=0$, so $p(L)=0$ and hence $f(z)\ge g(z)$ $p$-almost surely.
Similarly, $g(z)\ge f(z)$ $p$-almost surely, so $f(z)=g(z)$ $p$-almost surely.
\end{proof}

\begin{corollary}
\label{c-marg} Let $q$ be the marginal of $p$ on $X\times Z$.  Then, for each $J\in\mathcal{X}$, we have $p[J||\mathcal{Z}] = q[J||\mathcal{Z}]$ $q$-almost
surely.
\end{corollary}

\begin{lemma}
\label{l-01} If $p[J||\mathcal{Z}]\in\{0,1\}$ $p$-almost surely, then $p[J||\mathcal{Y}\otimes\mathcal{Z}] = p[J||\mathcal{Z}]$ $p$-almost surely.
\end{lemma}

\begin{proof}
Let $L_{0} = \{z\in Z:p[J||\mathcal{Z}]_{z}=0\}$ and $L_{1} = \{z\in Z:p[J||\mathcal{Z}]z=1\}$.  Then $L_{0},L_{1}\in\mathcal{Z}$ and $p(L_{0}\cup L_{1}) = 1$.  By Lemma \ref{l-cond},
\begin{equation*}
\int_{L_{0}}p[J||\mathcal{Z}]_{z}\,dp = 0 = p(J\times L_{0}),
\end{equation*}
\begin{equation*}
\int_{L_{1}}p[J||\mathcal{Z}]_{z}\,dp = p(L_{1}) = p(J\times L_{1}).
\end{equation*}
By Lemma \ref{l-cond} again,
\begin{equation*}
\int_{Y\times L_{0}}p[J||\mathcal{Y}\otimes\mathcal{Z}]_{(y,z)}\,dp = p(J\times Y\times L_{0}) = p(J\times L_{0}) = 0,
\end{equation*}
so
\begin{equation*}
p[J||\mathcal{Y}\otimes\mathcal{Z}]_{(y,z)} = 0 = p[J||\mathcal{Z}]_{z}\,\,\forall\,\,(y,z)\in Y\times L_{0}.
\end{equation*}
Similarly,
\begin{equation*}
\int_{Y\times L_{1}}p[J||\mathcal{Y}\otimes\mathcal{Z}]_{(y,z)}\,dp = p(J\times Y\times L_{1}) = p(J\times L_{1}) = p(L_{1}),
\end{equation*}
so
\begin{equation*}
p[J||\mathcal{Y}\otimes\mathcal{Z}]_{(y,z)} = 1 = p[J||\mathcal{Z}]_{z}\,\,\forall\,\,(y,z)\in Y\times L_{1},
\end{equation*}
as required.
\end{proof}
\medskip

When $x\in X$, we write $p[x||\mathcal{Z}]_{z}=p[\{x\}||\mathcal{Z}]_{z}$.  For the particular case of finite $X$, we get, by the properties of probability measures, that $\sum_{x\in X}p[x||\mathcal{Z}]_{z} = 1$ $p$-almost surely.

Given probability measures $p$ on $(X,\mathcal{X})\otimes (Y,\mathcal{Y})$ and $r$ on $(Y,\mathcal{Y})$, we say that $p$ is an \textbf{extension} of $r$ if $r = {\rm marg}_{Y}p$.  We say that two probability measures $p$ and $q$ on $(X,\mathcal{X})\otimes (Y,\mathcal{Y})$ \textbf{agree on} $Y$ if ${\rm marg}_{Y}p = {\rm marg}_{Y}q$.

Given probability spaces $(X,\mathcal{X},q)$ and $(Y,\mathcal{Y},r)$, the product measure $p=q\otimes r$ is the unique probability measure $p$ on $(X,\mathcal{X})\otimes (Y,\mathcal{Y})$ such that $q$ and $r$ are independent with respect to $p$, that is,
\begin{equation*}
p(J\times K)=q(J)\times r(K)
\end{equation*}
for all $J\in \mathcal{X}$ and $K\in \mathcal{Y}$.  Note that $p$ is a common extension of $q$ and $r$.

\begin{remark}
\label{r-product}
Let $(X,\mathcal{X},q)$ and $(Y,\mathcal{Y},r)$ be as above and let $p$ be a common extension of $q$ and $r$ on $(X,\mathcal{X})\otimes (Y,\mathcal{Y})$.  The following are equivalent:

\begin{enumerate}
\item[(i)] $p=q\otimes r$.

\item[(ii)]
 The $\sigma $-algebras $\mathcal{X}\otimes \{Y,\emptyset \}$ and $\{X,\emptyset \} \otimes \mathcal{Y}$ are independent with respect to $p$, that is,
\begin{equation*}
p(J\times K)=p(J)\times p(K)
\end{equation*}
for all $J\in \mathcal{X}$ and $K\in \mathcal{Y}$.

\item[(iii)] $p[J||\mathcal{Y}]_{y}=p(J)$ $p$-almost surely for all $J\in \mathcal{X}$.
\end{enumerate}
\end{remark}

We next introduce the notion of a fiber product of measures.  For the remainder of this section we let $\mathbf{X}=(X,\mathcal{X}),\mathbf{Y}=(Y,\mathcal{Y}), \mathbf{Z}=(Z,\mathcal{Z})$ be measurable spaces.

\begin{definition}
\label{d-fiber}
Let $q$ and $r$ be probability measures on $\mathbf{X}\otimes \mathbf{Z}$ and $\mathbf{Y}\otimes \mathbf{Z}$, respectively.  Assume that $q$ and $r$ have the same marginal $s$ on $\mathbf{Z}$.  We say that a probability measure $p$ on $\mathbf{X}\otimes \mathbf{Y}\otimes \mathbf{Z}$ is a \textbf{fiber product} of $q$ and $r$ over $Z$, in symbols $p = q\otimes _{Z}r$, if
\begin{equation*}
p(J \times K \times L) = \int_{L}q[J||\mathcal{Z}]_{z}\times r[K||\mathcal{Z}]_{z}\,ds
\end{equation*}
for all $J \in \mathcal{X}$, $K \in \mathcal{Y}$, and $L \in \mathcal{Z}$.
\end{definition}

Intuitively, the fiber product $q\otimes _{Z}r$ is the common extension of $q$ and $r$ with respect to which $q$ and $r$ are as independent as possible given that they have the same marginal on $Z$.  There are examples where a fiber product does not exist (see Swart \cite[1996]{swart96}).  But it is easily seen that if a fiber product $q\otimes _{Z}r$ does exist, then it is unique.  Next is a characterization of the fiber product in terms of conditional probabilities and extensions.

\begin{lemma}
\label{l-fiber}
Let $q$ and $r$ be as in Definition \ref{d-fiber}, and let $p$ be a common extension of $q,r$ on $\mathbf{X}\otimes \mathbf{Y}\otimes \mathbf{Z}$.  Then the following are equivalent:
\begin{enumerate}
\item[(i)] $p=q\otimes _{Z}r$.

\item[(ii)] $p[J\times K||\mathcal{Z}]_{z}=q[J||\mathcal{Z}]_{z}\times r[K||\mathcal{Z}]_{z}$ $p$-almost surely, for all $J\in \mathcal{X}$ and $K\in \mathcal{Y}$.

\item[(iii)] $p[J\times K||\mathcal{Z}]_{z}=p[J||\mathcal{Z}]_{z}\times p[K||\mathcal{Z}]_{z}$ $p$-almost surely, for all $J\in \mathcal{X}$ and $K\in \mathcal{Y}$.

\item[(iv)] $p[J||\mathcal{Y}\otimes \mathcal{Z}]_{(y,z)}=p[J||\mathcal{Z}]_{z}$ $p$-almost surely, for all $J\in \mathcal{X}$.
\end{enumerate}
\end{lemma}

\begin{proof}
It is clear that (i), (ii), and (iii) are equivalent.  Consider any $J\in \mathcal{X},K\in \mathcal{Y}$, and $L\in \mathcal{Z}$.  Assume (i).  To
prove (iv), it is enough to show that
\begin{equation*}
\int_{K\times L}p[J||\mathcal{Z}]\,dp = p(J\times K\times L).
\end{equation*}
We have
\begin{equation*}
\int_{K\times L}p[J||\mathcal{Z}]\,dp = \int_{Y\times L}p[J||\mathcal{Z}] \times 1_{K}\,dp.
\end{equation*}
By the rules of conditional expectations,
\begin{equation*}
{\rm E}[p[J||\mathcal{Z}]\times 1_{K}||\mathcal{Z}] = p[J||\mathcal{Z}] \times {\rm E}[1_{K}||\mathcal{Z}] = p[J||\mathcal{Z}] \times p[K||\mathcal{Z}].
\end{equation*}
Therefore
\begin{equation*}
\int_{Y\times L}p[J||\mathcal{Z}] \times 1_{K}\,dp = \int_{L}p[J||\mathcal{Z}] \times p[K||\mathcal{Z}]\,dp = \int_{L}q[J||\mathcal{Z}] \times r[K||\mathcal{Z}]\,dp.
\end{equation*}
By (i), this is equal to $p(J\times K\times L)$, which shows that (i) implies (iv).

Now assume (iv).  Then
\begin{equation*}
p(J\times K\times L) = \int_{K\times L}p[J||\mathcal{Y}\otimes \mathcal{Z}]\,dp = \int_{K\times L}p[J||\mathcal{Z}]\,dp = \int_{Y\times L}p[J||\mathcal{Z}]\times 1_{K}\,dp.
\end{equation*}
As in the preceding paragraph,
\begin{equation*}
\int_{Y\times L}p[J||\mathcal{Z}]\times 1_{K}\,dp = \int_{L}q[J||\mathcal{Z}]\times r[K||\mathcal{Z}]\,dp,
\end{equation*}
and condition (i) is proved.
\end{proof}
\medskip

A version $g(J, z)$ of the conditional probability $q[J||\mathcal{Z}]_z$ is \textbf{regular} if $g(\cdot, z_0)$ is a probability measure on $\mathbf{X}$ for each fixed $z_0 \in Z$.  It is well known that when $\mathbf{X}$ and $\mathbf{Z}$ are both Polish spaces, then $q[J||\mathcal{Z}]_z$ has a regular version.  It is also easily seen that when $X$ is finite and $\mathbf{Z}$ is any measurable space, then $q[J||\mathcal{Z}]_z$ has a regular version.  This is the case we will need in this paper.  The next lemma is from Swart \cite[1996]{swart96}:

\begin{lemma}
Let $q$ and $r$ be as in Definition \ref{d-fiber}.  If $q[J||\mathcal{Z}]_z$ has a regular version, then the fiber product $q\otimes _{Z}r$ exists.
\end{lemma}

\begin{corollary}
Let $q$ and $r$ be as in Definition \ref{d-fiber}.  If the space $X$ is finite, then the fiber product $q\otimes _{Z}r$ exists.
\end{corollary}

\section {Properties of Hidden-Variable Models}
\label {s-properties}
We can now formulate the various properties of \hv models which we listed in the Introduction (we will not repeat their sources) and establish some relationships among them.  At this point, we adopt:
\smallskip

\noindent \textbf{Assumption:} \textit{The outcome spaces $X_a$ and $X_b$ are finite, and $\mathcal{X}_a$ and $\mathcal{X}_b$ are the respective power sets.}
\medskip

\noindent Also, whenever we write an equation involving conditional probabilities, it will be understood to mean that the equation holds $p$-almost surely.  By the term \textquotedblleft measure\textquotedblright\,\,we will always mean \textquotedblleft probability measure.\textquotedblright\,\, Fix an \hv model $p$.  We will often make use of the following notation:
\begin{equation*}
p_a = {\rm marg}_{X_a \times Y \times \Lambda}p,\quad p_b = {\rm marg}_{X_b \times Y \times \Lambda}p,
\end{equation*}
\begin{equation*}
q_a = {\rm marg}_{X_a \times Y_a \times \Lambda}p,\quad q_b = {\rm marg}_{X_b \times Y_b \times \Lambda}p,
\end{equation*}
\begin{equation*}
r = {\rm marg}_{Y \times \Lambda}p,
\end{equation*}
\begin{equation*}
p_Y = {\rm marg}_{Y}p,\quad p_\Lambda = {\rm marg}_{\Lambda}p.
\end{equation*}

\begin{center}
\setlength{\unitlength}{0.8mm}
\begin{picture}(180,45)

\put(0,0){\line(1,0){20}}
\put(30,0){\line(1,0){20}}
\put(60,0){\line(1,0){20}}
\put(90,0){\line(1,0){20}}
\put(0,10){\line(1,0){20}}
\put(30,10){\line(1,0){20}}
\put(60,10){\line(1,0){20}}
\put(90,10){\line(1,0){20}}
\put(120,10){\line(1,0){20}}
\put(0,20){\line(1,0){20}}
\put(30,20){\line(1,0){20}}
\put(60,20){\line(1,0){10}}
\put(90,20){\line(1,0){20}}
\put(120,20){\line(1,0){20}}
\put(0,30){\line(1,0){20}}
\put(30,30){\line(1,0){10}}
\put(60,30){\line(1,0){10}}
\put(0,0){\line(0,1){30}}
\put(30,0){\line(0,1){30}}
\put(60,0){\line(0,1){30}}
\put(90,0){\line(0,1){20}}
\put(120,10){\line(0,1){10}}
\put(20,0){\line(0,1){30}}
\put(50,0){\line(0,1){20}}
\put(80,0){\line(0,1){10}}
\put(110,0){\line(0,1){20}}
\put(140,10){\line(0,1){10}}
\put(10,10){\line(0,1){20}}
\put(40,10){\line(0,1){20}}
\put(70,10){\line(0,1){20}}
\put(100,10){\line(0,1){10}}
\put(130,10){\line(0,1){10}}

\put(10,5){\makebox(0,0){$\Lambda$}}
\put(40,5){\makebox(0,0){$\Lambda$}}
\put(70,5){\makebox(0,0){$\Lambda$}}
\put(100,5){\makebox(0,0){$\Lambda$}}
\put(5,15){\makebox(0,0){$ Y_a$}}
\put(35,15){\makebox(0,0){$ Y_a$}}
\put(65,15){\makebox(0,0){$ Y_a$}}
\put(95,15){\makebox(0,0){$ Y_a$}}
\put(125,15){\makebox(0,0){$ Y_a$}}
\put(15,15){\makebox(0,0){$ Y_b$}}
\put(45,15){\makebox(0,0){$ Y_b$}}
\put(105,15){\makebox(0,0){$ Y_b$}}
\put(135,15){\makebox(0,0){$ Y_b$}}
\put(5,25){\makebox(0,0){$ X_a$}}
\put(35,25){\makebox(0,0){$ X_a$}}
\put(65,25){\makebox(0,0){$ X_a$}}
\put(15,25){\makebox(0,0){$ X_b$}}

\put(10,35){\makebox(0,0){$p$}}
\put(40,35){\makebox(0,0){$p_a$}}
\put(70,35){\makebox(0,0){$q_a$}}
\put(100,35){\makebox(0,0){$r$}}
\put(130,35){\makebox(0,0){$p_Y$}}

\put(150,0){\line(1,0){20}}
\put(150,0){\line(0,1){10}}
\put(150,10){\line(1,0){20}}
\put(170,0){\line(0,1){10}}
\put(160,5){\makebox(0,0){$\Lambda$}}
\put(160,35){\makebox(0,0){$p_\Lambda$}}

\end{picture}
\end{center}

All expressions below which are given for Alice have counterparts for Bob, with $a$ and $b$ interchanged.

\begin{definition}
The \hv model $p$ satisfies \textbf{locality} if for every $x\in X$ we have
\begin{equation*}
p[x||\mathcal{Y}\otimes \mathcal{L}] = p[x_{a}||\mathcal{Y}_{a}\otimes\mathcal{L}]\times p[x_{b}||\mathcal{Y}_{b}\otimes\mathcal{L}].
\end{equation*}
\end{definition}

\begin{definition}
The \hv model $p$ satisfies \textbf{parameter independence} if for every $x_{a}\in X_{a}$ we have
\begin{equation*}
p[x _{a}||\mathcal{Y}\otimes\mathcal{L}] = p[x _{a}||\mathcal{Y}_{a}\otimes\mathcal{L}].
\end{equation*}
\end{definition}

Here is a characterization of parameter independence in terms of fiber products.

\begin{corollary}
\label{c-fiber-parameter}
$p$ satisfies parameter independence if and only if $p_{a}=q_{a}\otimes _{Y _{a}\times \Lambda }r$ and $p_{b}=q_{b}\otimes_{Y_{b}\times \Lambda}r$.
\end{corollary}

\begin{proof}
By Lemma \ref{l-fiber}, $p_{a}=q_{a}\otimes _{Y_{a}\times \Lambda }r$ if and only if
\begin{equation*}
p_{a}[x_{a}||\mathcal{Y}\otimes \mathcal{L}]=p_{a}[x_{a}||\mathcal{Y}_{a}\otimes \mathcal{L}]
\end{equation*}
for all $x_{a}\in X_{a}$.  Since $p$ is an extension of $p_{a}$, this holds if and only if
\begin{equation*}
p[x_{a}||\mathcal{Y}\otimes \mathcal{L}]=p[x_{a}||\mathcal{Y}_{a}\otimes \mathcal{L}]
\end{equation*}
for all $x_{a}\in X_{a}$.  Similarly, $p_{b}=q_{b}\otimes _{Y_{b}\times \Lambda }r$ if and only if
\begin{equation*}
p[x_{b}||\mathcal{Y}\otimes \mathcal{L}]=p[x_{b}||\mathcal{Y}_{b}\otimes \mathcal{L}]
\end{equation*}
for all $x_{b}\in X_{b}$.  The result follows.
\end{proof}
\medskip

\begin{definition}
The \hv model $p$ satisfies \textbf{outcome independence} if for every $x=(x_a,x_b)\in X$ we have
\begin{equation*}
p[x||\mathcal{Y}\otimes\mathcal{L}] = p[x _{a}||\mathcal{Y}\otimes\mathcal{L}]\times p[x_{b}||\mathcal{Y}\otimes\mathcal{L}].
\end{equation*}
\end{definition}

The following corollary characterizes outcome independence in terms of fiber products.

\begin{corollary}
\label{c-fiber-outcome}
$p$ satisfies outcome independence if and only if $p = p_{a}\otimes _{Y \times \Lambda }p_{b}$.
\end{corollary}

\begin{proof}
This follows easily from Lemma \ref{l-fiber}.
\end{proof}
\medskip

The next proposition follows Jarrett \cite[1984, p.582]{jarrett84}.

\begin{proposition}
\label{p-param-outcome} $p$ satisfies locality if and only if it satisfies parameter independence and outcome independence.
\end{proposition}

\begin{proof}
It is easily seen from the definitions that if $p$ satisfies parameter independence and outcome independence, then $p$ satisfies locality.

Suppose that $p$ satisfies locality.  We have
\begin{equation*}
\{x_{a}\}\times X_{b} = \bigcup_{x_{b}\in X_{b}}\{(x_{a}, x_{b})\},
\end{equation*}
so
\begin{equation*}
p[x_{a}||\mathcal{Y}\otimes \mathcal{L}] = p[\{x_{a}\}\times X_{b}||\mathcal{Y}\otimes\mathcal{L}] = \sum_{x_{b}\in X_{b}}p[x_{a}, x_{b}||\mathcal{Y}\otimes \mathcal{L}] =
\end{equation*}
\begin{equation*}
\sum_{x_{b}\in X_{b}}(p[x_{a}||\mathcal{Y}_{a}\otimes\mathcal{L}]\times p[x_{b}||\mathcal{Y}_{b}\otimes\mathcal{L}]) =
\end{equation*}
\begin{equation*}
p[x_{a}||\mathcal{Y}_{a}\otimes\mathcal{L}]\times\sum_{x_{b}\in X_{b}}p[x_{b}||\mathcal{Y}_{b}\otimes\mathcal{L}] = p[x _{a}||\mathcal{Y}_{a}\otimes\mathcal{L}]\times 1 = p[x_{a}||\mathcal{Y}_{a}\otimes \mathcal{L}].
\end{equation*}
Similarly,
\begin{equation*}
p[x_{b}||\mathcal{Y}\otimes\mathcal{L}] = p[x_{b}||\mathcal{Y}_{b}\otimes\mathcal{L}].
\end{equation*}
It follows that $p$ satisfies parameter independence.

Again, supposing that $p$ satisfies locality, we have
\begin{equation*}
p[x_{a}, x_{b}||\mathcal{Y}\otimes\mathcal{L}] = p[x _{a}||\mathcal{Y}_{a}\otimes\mathcal{L}]\times p[x_{b}||\mathcal{Y}_{b}\otimes \mathcal{L}],
\end{equation*}
and hence
\begin{equation*}
p[x_{a}, x_{b}||\mathcal{Y}\otimes\mathcal{L}] = p[x _{a}||\mathcal{Y}\otimes \mathcal{L}]\times p[ x _{b}||\mathcal{Y}\otimes \mathcal{L}],
\end{equation*}
so $p$ satisfies outcome independence.
\end{proof}
\medskip

We immediately get a characterization of locality in terms of fiber products.

\begin{corollary}
\label{c-local}  $p$ satisfies locality if and only if
\begin{equation*}
p=p_{a}\otimes _{Y \times \Lambda}p_{b},\quad p_{a}=q_{a}\otimes_{Y_{a}\times \Lambda}r,\quad p_{b}=q_{b}\otimes _{Y_{b}\times\Lambda}r.
\end{equation*}
\end{corollary}

\begin{proof}
By Proposition \ref{p-param-outcome} and Corollaries \ref{c-fiber-parameter} and \ref{c-fiber-outcome}.
\end{proof}

\begin{definition}
The \hv model $p$ satisfies $\lambda$\textbf{-independence} if for every event $L\in\mathcal{L}$,
\begin{equation*}
p[L||\mathcal{Y}]_{y} = p(L).
\end{equation*}
\end{definition}

\begin{remark}
We observe:
\begin{enumerate}
\item[(i)] The $\lambda$-independence property for $p$ depends only on $r$.
\item[(ii)] Any \hv model $p$ such that $\Lambda$ is a singleton satisfies $\lambda$-independence.
\end{enumerate}
\end{remark}

By Remark \ref{r-product}, we have:

\begin{lemma}
\label{l-indep}
The following are equivalent:

\begin{enumerate}
\item[(i)]  $p$ satisfies $\lambda$-independence.

\item[(ii)] The measure $r$ is the product $p_Y \otimes p_\Lambda$.

\item[(iii)] The $\sigma$-algebras $\mathcal{Y}$ and $\mathcal{L}$ are independent with respect to $p$, i.e.,
\begin{equation*}
p(K\times L) = p(K)\times p(L)
\end{equation*}
for every $K\in\mathcal{Y},L\in\mathcal{L}$.
\end{enumerate}
\end{lemma}

The distinction between strong and weak determinism in the next two definitions is from Brandenburger and Yanofsky \cite[2008]{brandenburger-yanofsky08}.  Strong determinism is the notion discussed in the Introduction.

\begin{definition}
The \hv model $p$ satisfies \textbf{strong determinism} if for each $x_{a}\in X_{a}$ we have
\begin{equation*}
p[x _{a}||\mathcal{Y}_{a}\otimes \mathcal{L}]_{(y _{a},\lambda)}\in\{0,1\}.
\end{equation*}
\end{definition}

\noindent This says that the set $Y_{a}\times\Lambda$ can be partitioned into sets $\{A_{x_{a}}: x_{a}\in X _{a}\}$ such that $p[x_{a}||A_{x _{a}}] = 1$ for each $x _{a}\in X_{a}$.

\begin{definition}
The \hv model $p$ satisfies \textbf{weak determinism} if for each $x\in X$ we have
\begin{equation*}
p[ x||\mathcal{Y}\otimes \mathcal{L}]_{(y ,\lambda)}\in\{0,1\}.
\end{equation*}
\end{definition}

\noindent This says that the set $Y\times \Lambda$ can be partitioned into sets $\{A_{x}:x\in X\}$ such that $p[x||A_{x}] = 1$ for each $x\in X$.

\begin{lemma}
\label{l-weak-determ} The following are equivalent:

\begin{enumerate}
\item[(i)]  $p$ satisfies weak determinism.

\item[(ii)] For each $x_{a}\in X_{a}$ we have
\begin{equation*}
p[x _{a}||\mathcal{Y}\otimes \mathcal{L}]_{(y ,\lambda )}\in \{0,1\}.
\end{equation*}
\end{enumerate}
\end{lemma}

\begin{proof}
It is clear that (ii) implies (i).

Assume (i).  Then for $p$-almost all $(y ,\lambda)$ there is an $x\in X$ such that $p[x||\mathcal{Y}\otimes\mathcal{L}]_{(y ,\lambda)}=1$, and hence
\begin{equation*}
p[x _{a}||\mathcal{Y}\otimes\mathcal{L}]_{(y ,\lambda)} = 1
\end{equation*}
for each $x_{a}\in X_{a}$.  Therefore (ii) holds.
\end{proof}

\begin{proposition}
\label{p-determ} If $p$ satisfies strong determinism then it satisfies weak determinism.
\end{proposition}

\begin{proof}
Suppose $p$ satisfies strong determinism.  By Lemma \ref{l-01}, we have
\begin{equation*}
p[x _a||\mathcal{Y}_a\otimes\mathcal{L}] = p[x _a||\mathcal{Y}\otimes\mathcal{L}]
\end{equation*}
$p$-almost surely, and therefore
\begin{equation*}
p[x _a||\mathcal{Y}\otimes\mathcal{L}]\in \{0,1\},
\end{equation*}
so $p$ satisfies weak determinism by Lemma \ref{l-weak-determ}(ii).
\end{proof}

\begin{proposition}
\label{l-determ-outcome} If $p$ satisfies weak determinism then it satisfies outcome independence.
\end{proposition}

\begin{proof}
Suppose $p$ satisfies weak determinism.  By Lemma \ref{l-weak-determ}, we have
\begin{equation*}
p[x _{a}||\mathcal{Y}\otimes \mathcal{L}]\in\{0,1\}.
\end{equation*}
Therefore
\begin{equation*}
p[x||\mathcal{Y}\otimes\mathcal{L}] = p[x _{a}||\mathcal{Y}\otimes\mathcal{L}]\times p[x_{b}||\mathcal{Y}\otimes\mathcal{L}],
\end{equation*}
as required.
\end{proof}

\begin{proposition}
\label{l-strong-det-param}  $p$ satisfies strong determinism if and only if it satisfies weak determinism and parameter independence.
\end{proposition}

\begin{proof}
Suppose $p$ satisfies strong determinism.  By Lemma \ref{l-01},
\begin{equation*}
p[x _{a}||\mathcal{Y}_{a}\otimes\mathcal{L}] = p[x _{a}||\mathcal{Y}\otimes\mathcal{L}],
\end{equation*}
so $p$ satisfies parameter independence.  By Proposition \ref{p-determ}, $p$ satisfies weak determinism.

For the converse, suppose $p$ satisfies weak determinism and parameter independence.  Fix  $x_{a}\in X_{a}$.  By weak determinism and Lemma \ref{l-weak-determ},
\begin{equation*}
p[x _{a}||\mathcal{Y}\otimes\mathcal{L}]_{(y ,\lambda )}\in\{0,1\}.
\end{equation*}
By parameter independence,
\begin{equation*}
p[x_{a}||\mathcal{Y}\otimes\mathcal{L}] = p[x_{a}||\mathcal{Y}_{a}\otimes\mathcal{L}].
\end{equation*}
Therefore
\begin{equation*}
p[x _{a}||\mathcal{Y}_{a}\otimes\mathcal{L}]_{(y ,\lambda )}\in\{0,1\},
\end{equation*}
so $p$ satisfies strong determinism.
\end{proof}
\medskip

\begin{corollary}  \label{p-strong-det-implies-locality}
 $p$ satisfies strong determinism if and only if it satisfies weak determinism and locality.
\end{corollary}
\begin{proof}
By Propositions \ref{p-param-outcome}, \ref{l-determ-outcome}, and \ref{l-strong-det-param}.
\end{proof}
\bigskip
\bigskip

We can summarize the properties we have considered and the relationships among them in the following Venn diagram.

\vbox{\hsize=\textwidth
\hbox to \textwidth{\hss
\includegraphics[scale=0.55]{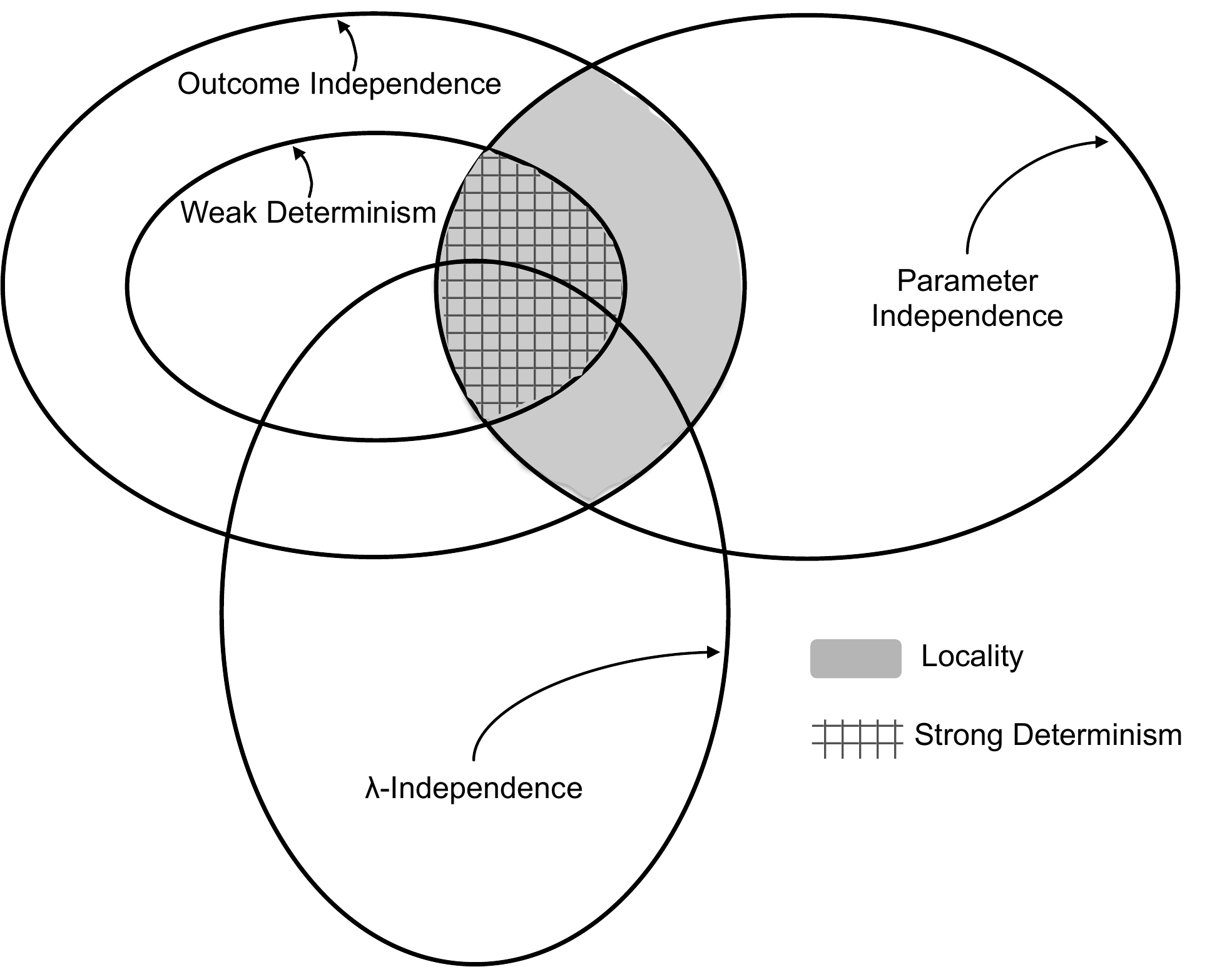}
\hss}
\vskip-20pt\relax
}
\medskip

\section{Determinization Theorems}
\label {s-firsttheorem}
Given an \hv model $p$, we call the probability space $(\Lambda ,\mathcal{L}, p_{\Lambda})$ the \textbf{\hv space of $p$}.

\begin{remark} \label{r-singleton}
Every empirical model $e$ can be realized by an \hv model $p$ where $p$ satisfies $\lambda$-independence
and the \hv space of $p$ has only one element.
\end{remark}

\begin{proof}  For every probability space $(\Lambda,\mathcal{L},p_\Lambda)$, the product measure $p=e\otimes p_\Lambda$ is an
\hv model that realizes $e$ and satisfies $\lambda$-independence.  In particular, we can take $\Lambda$ to be a one-element
set and take $(\Lambda,\mathcal{L},p_\Lambda)$ to be the trivial probability measure.
\end{proof}.
\medskip

We now state and prove our determinization results.

\begin{theorem}
\label{t-firsttheorem}
Every empirical model $e$ can be realized by an \hv model $p$ where $p$ satisfies strong determinism and the \hv space of $p$ is finite.
\end{theorem}

\begin{proof}
Let $s = {\rm marg}_{X}e$.  Build an \hv space $(\Lambda ,\mathcal{L}, s)$ where $\Lambda$ is a copy of $X$ and $\mathcal{L}$ is the power set of $X$.  Build a probability measure $d$ on $X \times \Lambda$ so that, for each $x \in X$ and $x^{\prime}\in \Lambda$,
\begin{equation*}
d(x, x^{\prime}) =
\begin{cases}
s(x) & \text{if $x = x^{\prime}$}, \\
0 & \text{otherwise}.
\end{cases}
\end{equation*}
\noindent Note that $d$ is an extension of $s$.

Let $p$ be the fiber product $p = d \otimes_{X} e$.  Then $p$ is realization-equivalent to $e$.  Since $\Lambda$ is a copy of the finite space $X$, $\Lambda$ is finite.  For each $x_a \in X_a$ and $x^{\prime} \in \Lambda$, we have
\begin{equation*}
p[x_{a} || x^{\prime}] = d[x_{a} || x^{\prime}] \in \{0, 1\}.
\end{equation*}
By Lemma \ref{l-01}, for each $x_a$ we have
\begin{equation*}
p[x_{a} || \mathcal{Y}_{a}\otimes \mathcal{L}]_{(y _{a}, x^{\prime})} \in \{0, 1\}
\end{equation*}
$p$-almost surely.  This shows that $p$ satisfies strong determinism.
\end{proof}

\begin{theorem}
\label{t-secondtheorem}
Given an \hv model $p$ satisfying locality and $\lambda$-independence, there is a realization-equivalent \hv model $\bar{p}$ that satisfies strong determinism and $\lambda$-independence.
\end{theorem}

\begin{proof}
Suppose $p$ satisfies locality and $\lambda $-independence.  We will construct a new \hv model $\bar{p}$ whose \hv space $(\bar{\Lambda},\bar{\mathcal{L}},\bar{p}_{\bar{\Lambda}})$ will be the product of $(\Lambda ,\mathcal{L},p_{\Lambda})$ and the Lebesgue unit square
\begin{equation*}
([0,1]_{a},\mathcal{U}_{a},u_{a})\otimes ([0,1]_{b},\mathcal{U}_{b},u_{b}).
\end{equation*}
Here, $[0,1]_a$ is a copy of the real unit interval, $\mathcal{U}_a$ is the set of Borel subsets of $[0,1]_a$, and $u_a$ is Lebesgue measure on $\mathcal{U}_a$; similarly for $b$.

Let $X_{a}=\{x_{a}^{1},\ldots ,x_{a}^{A}\}$.  For each $y_{a}\in Y_{a}$ and $\lambda \in \Lambda$, partition $[0,1]_{a}$ into $A$ consecutive intervals
\begin{equation*}
I_{a}(x_{a}^{1},y _{a},\lambda), I_{a}(x_{a}^{2},y _{a},\lambda), \ldots, I_{a}(x_{a}^{A},y _{a},\lambda),
\end{equation*}
where, for each $x_a\in X_a$, $I_{a}(x_{a},y _{a},\lambda)$ has length
\begin{equation*}
u_{a}(I_{a}(x_{a},y _{a},\lambda )) = p[x_{a}||\mathcal{Y}_{a}\otimes \mathcal{L}]_{(y _{a},\lambda)}.
\end{equation*}
Note that the boundary point between the $i$th and $(i+1)$th intervals is the $(\mathcal{Y}_{a}\otimes \mathcal{L})$-measurable function
\begin{equation*}
\sum_{i=1}^{n}p[x_{a}^{i}||\mathcal{Y}_{a}\otimes \mathcal{L}]_{(y _{a},\lambda)}.
\end{equation*}
We carry out the same construction with $b$ in place of $a$.

Let $\bar{r} = r \otimes u_{a}\otimes u_{b}$.  Since $p$ satisfies $\lambda $-independence, $r = p_{Y} \otimes p_{\Lambda}$, and thus $\bar{r} = p_{Y}\otimes \bar{p}_{\bar{\Lambda}}$.  Let $s_{a}$ be the unique probability measure on
\begin{equation*}
(X_{a},\mathcal{X}_{a}) \otimes (Y_{a},\mathcal{Y}_{a}) \otimes (\Lambda,\mathcal{L}) \otimes ([0,1]_{a},\mathcal{U}_{a})
\end{equation*}
such that for each $x_{a}\in X_{a}$, $K_{a}\in \mathcal{Y}_{a}$, $L\in \mathcal{L}$, and $U_{a}\in \mathcal{U}_{a}$, we have
\begin{equation*}
s_{a}(\{x_{a}\} \times K_{a}\times L\times U_{a}) = \int_{K_{a}\times L\times U_{a}}1_{I_{a}(x_{a},y_{a},\lambda)}(\alpha)\,d\bar{r},
\end{equation*}
where we write $\alpha $ for a typical element of $[0,1]_{a}$.  Define $s_{b}$ in a similar way.  Now define $\bar{p}_{a}$, $\bar{p}_{b}$, and $\bar{p}$ as the fiber products
\begin{equation*}
\bar{p}_{a} = s_{a}\otimes _{Y_{a}\times \Lambda \times \lbrack 0,1]_{a}} \bar{r} \text{, \ } \bar{p}_{b} = s_{b} \otimes _{Y_{b}\times \Lambda \times \lbrack 0,1]_{b}} \bar{r} \text{, \ } \bar{p} = \bar{p}_{a}\otimes_{Y\times \bar{\Lambda }} \bar{p}_{b}.
\end{equation*}
We see that the \hv model $\bar{p}$ is a common extension of $s_{a}$, $s_{b}$, and $\bar{r}$.  It also satisfies $\lambda$-independence because $\bar{r} = p_{Y}\otimes \bar{p}_{\bar{\Lambda}}$.  By Lemma \ref{l-cond},
\begin{equation*}
s_{a}[x_{a}||\mathcal{Y}_{a}\otimes \mathcal{L}\otimes \mathcal{U}_{a}] = 1_{I_{a}(x_{a},y_{a},\lambda)}\in \{0,1\}.
\end{equation*}
By Lemma \ref{l-01},
\begin{equation*}
s_{a}[x_{a}||\mathcal{Y}_{a}\otimes \bar{\mathcal{L}}] \in \{0,1\}.
\end{equation*}
Similarly for $s_{b}$.  Therefore $\bar{p}$ satisfies strong determinism.

It remains to prove that $\bar{p}$ is an extension of $p$.  By Fubini's Theorem,
\begin{equation*}
s_{a}(\{x_{a}\} \times K_{a}\times L) = \int_{K_{a}\times L\times \lbrack 0,1]_{a}} 1_{I_{a}(x_{a},y_{a},\lambda)}(\alpha )\,d\bar{r} =
\end{equation*}
\begin{equation*}
\int_{K_{a}\times L} \int_{0}^{1} 1_{I_{a}(x_{a},y_{a},\lambda)}(\alpha)\,du_{a}\,dr = \int_{K_{a}\times L} u_{a}(I_{a}(x_{a},y_{a},\lambda))\,dr =
\end{equation*}
\begin{equation*}
\int_{K_{a}\times L} q_{a}[x_{a}||\mathcal{Y}_{a}\otimes \mathcal{L}]_{(y_{a},\lambda)}\,dr = q_{a}(\{x_{a}\} \times K_{a}\times L).
\end{equation*}
Thus $s_{a}$ is an extension of $q_{a}$.  Similarly, $s_{b}$ is an extension of $q_{b}$.

Since both $p$ and $\bar{p}$ satisfy locality, and $\bar{p}$ extends $\bar{r} = r \otimes u_{a}\otimes u_{b}$, by Fubini's Theorem we have
\begin{equation*}
\bar{p}(\{x\} \times K\times L) = \int_{K \times L\times \lbrack 0,1]_{a}\times \lbrack 0,1]_{b}} \bar{p}[x||\mathcal{Y}\otimes \bar{\mathcal{L}}]\,d\bar{r} =
\end{equation*}
\begin{equation*}
\int_{K \times L\times \lbrack 0,1]_{a}\times \lbrack 0,1]_{b}} s_{a}[x_{a}||\mathcal{Y}_{a}\otimes \bar{\mathcal{L}}] \times s_{b}[x_{b}||\mathcal{Y}_{b}\otimes \bar{\mathcal{L}}]\,d\overline{r} =
\end{equation*}
\begin{equation*}
\int_{K\times L}\int_{0}^{1}\int_{0}^{1} s_{a}[x_{a}||\mathcal{Y}_{a}\otimes \bar{\mathcal{L}}]\times s_{b}[x_{b}||\mathcal{Y}_{b}\otimes \bar{\mathcal{L}}]\,du_{a}\,du_{b}\,dr =
\end{equation*}
\begin{equation*}
\int_{K\times L} q_{a}[x_{a}||\mathcal{Y}_{a}\otimes \mathcal{L}] \times q_{b}[x_{b}||\mathcal{Y}_{b}\otimes \mathcal{L}]\,dr = p(\{x\} \times K\times L).
\end{equation*}
Thus $\bar{p}$ is an extension of $p$, and hence $\bar{p}$ is realization-equivalent to $p$.  This completes the proof.
\end{proof}
\medskip

All the results in Section \ref{s-properties} (\lq\lq Properties of Hidden-Variable Models\rq\rq), and Theorems \ref{t-firsttheorem} and \ref{t-secondtheorem} in this section, extend immediately to multipartite systems.  The only adjustment needed is that parameter independence must now be stated in terms of sets of parts instead of individual parts.  Interestingly, outcome independence and locality do not need to be restated.

\section{Endnote}
\label{s-endnote}
To keep things simple, we assumed in this paper that the outcome spaces $X_a$ and $X_b$ are finite.  However, the only result in this paper that requires this assumption is Theorem \ref{t-firsttheorem}.  We show in \cite[2012]{brandenburger-keisler12} that all of the results in Section \ref{s-properties} hold for arbitrary outcome spaces $X_a$ and $X_b$, and that Theorem \ref{t-secondtheorem} holds assuming only that the outcome spaces have countably generated $\sigma$-algebras of events $\mathcal{X}_{a}$ and $\mathcal{X}_{b}$.

It would be of interest to extend the methods in this paper to formulate other properties that have usually been studied only for the case of finite sets of measurements.  For finite probability spaces, Abramsky and Brandenburger \cite[2011]{abramsky-brandenburger11} establish a strict hierarchy of three properties: non-locality (\`a la Bell) is strictly weaker than possibilistic non-locality (exhibited by the Hardy \cite[1993]{hardy93} model), which is strictly weaker than strong contextuality (exhibited by the Greenberger, Horne, and Zeilinger \cite[1989]{greenberger-horne-zeilinger89} model).  (In this language, the Kochen-Specker Theorem \cite[1967]{kochen-specker67} is a model-independent proof of strong contextuality.)  Extending these latter properties to the general measure-theoretic setting appears to be an open direction.

\end{document}